\documentclass[final,a4paper]{article}
\usepackage[T1]{fontenc}
\usepackage{array}
\usepackage{textcomp}
\usepackage{mathrsfs,amssymb,amsmath,amsthm}
\usepackage{graphicx}
\usepackage{hyperref}
\usepackage{pst-all}
\usepackage{mathpazo}

\newcommand{\N}{\mathbb{N}}

\newcommand{\R}{\mathbb{R}}

\newcommand{\Var}{\mathbb V\mathbf{ar}}
\newcommand{\E}{\mathbb{E}}
\newcommand{\Cov}{\mathbb C\mathbf{ov}}

\DeclareMathOperator*{\argmin}{argmin}
\providecommand{\norm}[1]{\left\lVert#1\right\rVert}

\providecommand{\abs}[1]{\left|#1\right|}
\everymath{\displaystyle\everymath{}}
\newtheorem{defn}{Definition}
\newtheorem{prop}{Proposition}
\newtheoremstyle{break}
  {\topsep}{\topsep}%
  {\itshape}{}%
  {\bfseries}{}%
  {\newline}{}%
\newtheorem{thm}{Theorem}
\newtheorem{cor}{Corollary}
\newtheorem{lem}{Lemma}
\newtheorem*{rem}{Remark}
\begin{document}

\title{{Randomized pick-freeze for sparse Sobol indices estimation in high dimension}}
\footnotetext[1]{Laboratoire de Math\'ematiques d'Orsay, B\^atiment 425, Universit\'e Paris-Sud, 91405 Orsay, France}
\author{\renewcommand{\thefootnote}{\arabic{footnote}}Yohann de Castro\footnotemark[1]{\ } and Alexandre Janon\footnotemark[1]}


\maketitle

\begin{abstract}
This article investigates a new procedure to estimate the influence of each variable of a given function defined on a high-dimensional space. More precisely, we are concerned with describing a function of a large number $p$ of parameters that depends only on a small number $s$ of them. Our proposed method is an unconstrained $\ell_{1}$-minimization based on the Sobol's method. We prove that, with only $\mathcal O(s\log p)$ evaluations of $f$, one can find which are the relevant parameters.
\end{abstract}

\section{Introduction}

\subsection{Context: Sensitivity analysis and Sobol indices}

Some mathematical models encountered in applied sciences involve a large number of poorly-known parameters as inputs. It is important for the practitioner to assess the impact of this uncertainty on the model output. An aspect of this assessment is sensitivity analysis, which aims to identify the most sensitive parameters, that is, parameters having the largest influence on the output. The parameters identified as influent have to be carefully tuned (or estimated) by the users of the model. On the other hand, parameters whose uncertainty has a small impact can be set to a nominal value (which can be some special value, for which the model is simpler).

In global (stochastic) variance-based sensitivity analysis (see for example \cite{saltelli2008global} and references therein), the input variables are assumed to be independent random variables. Their probability distributions account for the practitioner's belief about the input uncertainty. This turns the model output into a random variable, whose total variance can be split down into different partial variances (this is the so-called Hoeffding decomposition, also known as functional ANOVA, see \cite{liu2003estimating}). Each of these partial variances measures the uncertainty on the output induced by each input variable uncertainty. By considering the ratio of each partial variance to the total variance, we obtain a measure of importance for each input variable that is called the Sobol index or sensitivity index of the variable \cite{sobol1993,sobol2001global}; the most sensitive parameters can then be identified and ranked as the parameters with the largest Sobol indices. Each partial variance can be written as the variance of the conditional expectation of the output with respect to each input variable. 

Once the Sobol indices have been defined, the question of their effective computation or estimation remains open. In practice, one has to estimate (in a statistical sense) those indices using a finite sample (of size typically in the order of hundreds of thousands) of evaluations of model outputs \cite{tissot2012bias}. Indeed, many Monte Carlo or quasi Monte Carlo approaches have been developed by the experimental sciences and engineering communities. This includes the Sobol pick-freeze (SPF) scheme (see \cite{sobol2001global,gamboa2013statistical}). In SPF a Sobol index is viewed as the regression coefficient between the output of the model and its pick-freezed replication. This replication is obtained by holding the value of the variable of interest (frozen variable) and by sampling the other variables (picked variables). The sampled replications are then combined to produce an estimator of the Sobol index.

\subsection{High-dimensional, sparse contexts}
The pick-freeze scheme is used on models with a reasonable (typically, less than one thousand) number of inputs. When there is a large number of input parameters (what we call an \emph{high-dimensional} context), this scheme will require a number of model evaluations which is generally too large to be computed in practice. Hence, in high-dimensional contexts, some specific sensitivity analysis methods exist, such as \emph{screening} methods (for instance,  Morris' scheme \cite{morris1991factorial}), but they do not target the estimation of Sobol indices. Note that in \cite{tissot2012estimating}, an interesting method for estimating Sobol indices is proposed and is claimed to be applicable in high-dimensional contexts.

Besides, models with a large number of input parameters often display a so-called \emph{sparsity of effects} property, that is, only a small number of input parameters are actually influent: in other terms, we want to efficiently estimate a sparse vector of Sobol indices. Sparse estimation in high-dimensional contexts is the object of high-dimensional statistics methods, such as the LASSO estimator. 

In our frame, we would like to find the most influent inputs of a function that is to be described. This framework is closely related to exact support recovery in high-dimensional statistics. Note exact support recovery using $\ell_{1}$-minimization has been intensively investigated during the last decade, see for instance \cite{zhao2006model,fuchs2004sparse,tropp2006just,lounici2008sup} and references therein. We capitalize on these works to build our estimation procedure. The goal of this paper is to draw a bridge, which, to the best of our knowledge, has not been previously drawn, between Sobol index estimation via pick-freeze estimators and sparse linear regression models. This bridge can be leveraged so as to propose an efficient estimation procedure for Sobol indices in high-dimensional sparse models.

\subsection{Organization of the paper}
The contribution of this paper is twofold: Section \ref{sec:ConvexRelax} describes a new algorithm to simultaneously estimate the Sobol indices using $\ell_{1}$-relaxation and give elementary error analysis of this algorithm (Theorem \ref{t:linfini} and Theorem \ref{t:linfinirade}), and Section \ref{sec:RademacherBreaking} presents a new result on exact support recovery by Thresholded-Lasso that do not rely on coherence propriety. In particular, we prove that exact support recovery holds beyond the Welch bound. Appendix \ref{App:ProofA} gives the proofs of the results in Section \ref{sec:ConvexRelax}. Appendix \ref{App:WithoutWelch} gives preliminary results for proving Theorem \ref{thm:tiebreak} of Section \ref{sec:RademacherBreaking}. The remaining appendices apply these results to different designs (leading for Appendix \ref{App:RadeClassics} to Theorem \ref{thm:tiebreak}); Appendix \ref{App:AppendixGraph} and \ref{App:BerDesigns} are rather independent and study Thresholded-Lasso in the frame of random sparse graphs.

\section{A convex relaxation of Sobol's method}
\label{sec:ConvexRelax}
\subsection{Notation and model}
Denote by $X_1, \ldots, X_p$ the input parameters, assumed to be independent random variables of known distribution. Let $Y$ be the model output of interest:
\[ Y = f(X_1, \ldots, X_p), \]
where $f: \R^p \rightarrow \R$ is so that $Y \in L^2$ and $\Var(Y) \neq 0$. Assume that $f$ is \emph{additive}, i.e.
\begin{equation}\label{e:modadditif} f(X_1, \ldots, X_p) = f_1 (X_1) + \ldots + f_p(X_p) \end{equation}
for some functions $f_i: \R \rightarrow \R$, $i=1, \ldots, p$. We want to estimate the following vector:
\[ S = \left( S_i \right)_{i=1}^{p}\quad\mathrm{where}\quad S_i = \frac{\Var[ \E(Y|X_i)]}{\Var(Y)}, \]
is the $i^\text{th}$ Sobol index of $Y$ and quantifies the influence of $X_i$ on $Y$. In this article we present a new procedure for evaluating the Sobol indices when $p$ is large. We make the assumption that the number of nonzero Sobol indices:
\[ s = \# \{ i=1,\ldots,p\ \text{ s.t. }\ S_i \neq 0 \} \]
remains small in comparison to $p$. Observe our model assumes that we know an upper bound on $s$. The Sobol indices can be estimated using the so-called \emph{pick-freeze} scheme, also know as \emph{Sobol's method} \cite{sobol1993,sobol2001global}. Let $X'$ be an independent copy of $X$ and note, for $i=1,\ldots,p$:
\begin{equation}\label{e:defyi} Y^i = f(X_1', \ldots, X_{i-1}', X_i, X_{i+1}', \ldots, X_p'). \end{equation}
Then we have:
\[ S_i = \frac{\Cov(Y, Y^i)}{\Var (Y)}. \]
This identity leads to an empirical estimator of $S_i$:
\[ \widehat S_i = \frac{ \frac 1N \sum Y_k Y_k^i - \left( \frac{1}{N} \sum \frac{Y_k + Y_k^i}{2} \right)^2 }{ \frac 1N \sum \frac{(Y_k)^2 + (Y_k^i)^2}{2} - \left( \frac{1}{N} \sum \frac{Y_k + Y_k^i}{2} \right)^2 }, \]
where all sums are for $k$ from $1$ to $N$, and  $\{ (Y_k, Y_k^i) \}_{k=1,\ldots,N}$ is an iid sample of the distribution of $(Y, Y^i)$ of size $N$. This estimator has been introduced in \cite{Monod2006} and later studied in \cite{janon2012asymptotic} and \cite{gamboa2013statistical}.

In the high-dimensional frame, the estimation of the $p$ indices using $\widehat S_i$ for $i=1,\ldots,p$ would require $(p+1)N$ evaluations of $f$ so as to generate the realizations of $(Y, Y^1, \ldots, Y^p)$. This may be too much expensive when $p$ is large and/or evaluation of $f$ is costly. Besides, thanks to our sparsity assumption, such an estimation ``one variable at a time'' will be inefficient, as many computations will be required to estimate zero many times. To the best of our knowledge, this paper is the first to overcome this difficulty introducing a new estimation scheme.



\subsection{Multiple pick-freeze}
We now generalize definition \eqref{e:defyi}. Let $F \subset \{ 1, \ldots, p \}$ be a set of indices. Define $Y^F$ by:
\[ Y^F = f(X^F) \quad \text{ where } \quad \big(X^F\big)_i = \left\{ \begin{array}{l} X_i \text{ if } i \in F\,, \\ X_i' \text{ if } i \in F^{c}\,. \end{array} \right. \]
where $F^{c}=\{1,\ldots,p\}\setminus F$. The name of the method stems from the fact that, to generate the $Y^F$ variable, all the input parameters whose indices are in $F$ are \textbf{\underline F}rozen. In the pick-freeze scheme of the previous subsection, only one variable was frozen at the time, namely $F=\{i\}$. We then define:
\[ S_F = \frac{\Cov(Y, Y^F)}{\Var(Y)}, \]
which admits a natural estimator:
\begin{equation}\label{e:defestsf} \widehat S_F = \frac{ \frac 1N \sum Y_k Y_k^F - \left( \frac{1}{N} \sum \frac{Y_k + Y_k^F}{2} \right)^2 }{ \frac 1N \sum \frac{(Y_k)^2 + (Y_k^F)^2}{2} - \left( \frac{1}{N} \sum \frac{Y_k + Y_k^F}{2} \right)^2 }. \end{equation}
Under additivity hypothesis \eqref{e:modadditif}, one has:
\[ S_F = \sum_{i \in F} S_i. \]
Now, let's choose $n \in \N^\star$, subsets $F_1, \ldots, F_n$ of $\{1, \ldots, p\}$, and denote by $E$ the following vector of estimators:
\begin{equation}\label{e:defE} E = (\widehat S_{F_1}, \ldots ,\widehat S_{F_n} )\,. \end{equation}
Notice that, once the $F_1, \ldots, F_n$ have been chosen, the $E$ vector can be computed using $(n+1) N$ evaluations of $f$.

\subsubsection{Bernoulli Regression model}
The choice of $F_1, \ldots, F_n$ can be encoded in a binary matrix $\Phi$ with $n$ lines and $p$ columns, so that:
\begin{equation}
\label{e:codagephiBer}
\Phi_{ji} = \left\{ \begin{array}{l} 1 \text{ if } i \in F_j\,, \\ 0 \text{ otherwise. } \end{array} \right. \quad j=1,\ldots,n\ \mathrm{and} \ i=1,\ldots,p.  
\end{equation}
It is clear that $(S_{F_1} ,\ldots , S_{F_n})= \Phi S$, hence:
\begin{equation}\label{e:modlin} E = \Phi S + \epsilon, \end{equation}
where the $\epsilon$ vector defined by $\epsilon_j = \widehat S_{F_i} - S_{F_i}$ gives the estimation error of $\Phi S$ by $E$. In practice, the $E$ vector and the $\Phi$ matrix are known, and one has to estimate $S$. Thereby, Eq. \eqref{e:modlin} can be seen as a linear regression model whose coefficients are the Sobol indices to estimate. Moreover, observe that $n \ll p$ and $S$ sparse, hence we are in a high-dimensional sparse linear regression context. The problem \eqref{e:modlin} has been extensively studied in the context of sparse estimation \cite{lounici2008sup,bickel2009simultaneous,de2012remark} and compressed sensing \cite{candes2006stable,candes2009near}, and a classical solution is to use the LASSO estimator \cite{tibshirani1996regression}:
\begin{equation}\label{e:defSest} \widehat S \in \argmin_{U \in \R^p} \,
\left( \frac{1}{n} \norm{ E - \Phi U }_2^2 + 2r \norm{U}_1 \right)
, \end{equation}
where $r>0$ is a regularization parameter and:
\[ \norm{v}_2^2 = \sum_{j=1}^n v_j^2\,, \quad \norm{u}_1 = \sum_{i=1}^p |u_i|\,. \]
\noindent
Many efficient algorithms, such as LARS \cite{efron2004least}, are available in order to solve the above minimization problem, and to find an appropriate value for $r$. In high dimensional statistics, one key point for the LASSO procedure is the choice of the $\Phi$ matrix. In the Compressed Sensing literature, a random matrix with i.i.d. coefficients often proves to be a good choice, hence we will study possible random choices for $\Phi$.
\bigskip

\hrule

\subsubsection*{Summary of the method ``Randomized Pick-Freeze'' (RPF) for Bernoulli matrices}
\label{summaryber}
Our estimation method can be summarized as follows:
\begin{enumerate}
\item Choose $N$ (Monte-Carlo sample size), $n$ (number of estimations) and $r$ (regularization parameter).
\item Randomly sample a 0-1 matrix $\Phi$ with Bernoulli distribution of parameter $\mu$.
\item Deduce from $\Phi$ the $F_1, \ldots, F_n$ subsets using \eqref{e:codagephiBer}.
\item Generate a $N$-sized sample of $(Y, Y^{F_1}, \ldots, Y^{F_n})$.
\item Use this sample in \eqref{e:defestsf}, for $F=F_1, \ldots, F_n$, to obtain the $E$ vector \eqref{e:defE}.
\item Solve problem \eqref{e:defSest} to obtain an $\widehat S$ which estimates $S$.
\end{enumerate}

\hrule
\bigskip

\noindent
Given the binary constraint on $\Phi$, we will choose a Bernoulli distribution with parameter $\mu \in ]0;1[$. In this model, $(\Phi_{ji})_{j,i}$ are independent, with for all  $i,j$:
\begin{equation}
\label{e:loibern}
\mathbb P(\Phi_{ji} = 1)=\mu=1-\mathbb P(\Phi_{ji}=0). 
\end{equation}
\begin{thm}[$\ell^\infty$ error bound] 
\label{t:linfini}
Suppose that:
\begin{enumerate}
\item $\delta$ is a real in $\left]0 ; \frac{1-\mu}{16s}\right[$;
\item $\epsilon$ is a centered Gaussian vector whose covariance matrix has $\sigma^2$ as largest eigenvalue;
\item $r=A\sigma\sqrt{\mu(1+\delta)}\sqrt{\frac{\ln p}{n}}$ for some $A>2\sqrt 2$. 
\end{enumerate}
Let:
\begin{eqnarray*}
t &=& \left( \frac 3 2 + \frac{ 24 (\mu+\delta) }{ \frac{1-\mu}{s} - 16\delta } \right) \frac{r}{\mu}\,; \\
\alpha&=&1 - \left( 1- p^{1-A^2/8} \right) \left( 1 -   2 \exp \left( - 2 n \delta^2 \mu^2 + \ln p \right) \right) \\&&+ \exp \left( - 2 n \delta^2 \mu^2 + 2 \ln p \right). 
\end{eqnarray*}
Then, with probability at least $1-\alpha$, any solution $\widehat S$ of \eqref{e:defSest} satisfies:
\[
 \max_{i=1,\ldots,p} | \widehat S_i - S_i | \leq t.
\]
\end{thm}
\begin{proof}
A proof can be found in Appendix \ref{proof:TheoBer}.
\end{proof}

\noindent
\textbf{Remark 1: } For the probability above to be greater than zero, it is necessary to have:
\begin{equation}
\label{e:FactureBernoulli1}
n \geq \frac{\ln p}{\delta^2 \mu^2} \geq \frac{ 256 s^2 \ln p }{\mu^2 (1-\mu)^2}\,. 
\end{equation}
\noindent
\textbf{Remark 2: } The statement of Theorem \ref{t:linfini} can be compared to standard results in high-dimensional statistics such as exact support recovery under coherence property \cite{lounici2008sup}. Nevertheless, observe that a standard assumption is that the column norm of the design matrix is $\sqrt n$ while in our frame this norm is random with expectation of order $\sqrt{\mu n}$.

\noindent
\textbf{Remark 3: } In our context, the second hypothesis of the above theorem does not exactly hold; indeed, $\epsilon$ is only asymptotically Gaussian (when $N \rightarrow +\infty$), see \cite{janon2012asymptotic} for instance. However, for practical purposes, the observed convergence is fast enough. One can also see a related remark in our proof of this theorem, in Appendix \ref{proof:TheoBer}. 

\begin{cor}[Support recovery by Thresholded-Lasso]
\label{c:supp}
\label{c:supprade}
Let: 
\[ S_{min} = \min_{\substack{i=1,\ldots,p \\ \text{s.t. }S_i \neq 0}} S_i. \]
\noindent
Then, under the same assumptions of Theorem \ref{t:linfini}, we have, with probability greater than $1-\alpha$ and for all $i=1,\ldots,p$:
\[ \widehat S_i > t \;\Longrightarrow\; S_i > 0, \]
and:
\[ \widehat S_i < S_{min} - t \;\Longrightarrow\; S_i = 0. \]
\end{cor}

\begin{proof}[Proof of Corollary \ref{c:supp}]
For the first point, notice that:
\[ |S_i| \geq |\widehat S_i| - |\widehat S_i - S_i| \geq |\widehat S_i| - t > 0 \text{ if } \widehat S_i>t. \]
For the second point: if $\widehat S_i < S_{min}-t$, we have:
\[ |S_i| \leq |S_i - \widehat S_i| + |\widehat S_i| < t + (S_{min}-t) = S_{min}, \]
and $S_i=0$ by definition of $S_{min}$.
\end{proof}

\noindent
\textbf{Remark 4 (important): } Theorem \ref{t:linfini} and Corollary \ref{c:supp} show that one can identify the most important inputs of a function as soon as the corresponding Sobol indices are above the threshold $t$. Recall \textit{Thresholed-Lasso} is a thresholded version of any solution to \eqref{e:defSest}. Moreover, observe that we do not address the issue of estimating the Sobol indices. This can be done using a two-step procedure: estimate the support using Thresholed-Lasso and then estimate the Sobol indices using a standard least squares estimator.

\subsubsection{Rademacher Regression model}
The choice of $F_1, \ldots, F_n$ can also be encoded in a $\pm1$ matrix $\Phi$ with $n$ lines and $p$ columns, so that:
\begin{equation}
\label{e:codagephi}
\Phi_{ji} = \left\{ \begin{array}{ll} 1 & \text{ if } i \in F_j\,, \\ -1 &\text{ otherwise. } \end{array} \right. \quad j=1,\ldots,n\ \mathrm{and} \ i=1,\ldots,p.  
\end{equation}
\noindent
It is clear that:
\[ (S_{F_1} ^{\Delta},\ldots , S_{F_n}^{\Delta})= \Phi S\,, \]
where $S_{F_i} ^{\Delta}=S_{F_{i}}-S_{F_{i}^{c}}$. Hence:
\begin{equation}\label{e:modlinRade} E = \Phi S + \epsilon^{\Delta}, \end{equation}
where the $\epsilon$ vector defined by $\epsilon_j^{\Delta} = \widehat S_{F_i}^{\Delta} - S_{F_i}^{\Delta}$ gives the estimation error of $\Phi S$ by $E$. Thus, the problem of estimating $S$ from $E$ has been casted into linear regression which can be tackled by \eqref{e:defSest}.

\bigskip

\hrule

\subsubsection*{Summary of the method ``Randomized Pick-Freeze'' (RPF) for Rademacher matrices}
\label{summaryrad}
Our estimation method can be summarized as follows:
\begin{enumerate}
\item Choose $N$ (Monte-Carlo sample size), $n$ (number of estimations), and $r$ (regularization parameter).
\item Sample a $\Phi$ matrix according to a $\pm1$ symmetric Rademacher distribution.
\item Deduce from $\Phi$ the $F_1, \ldots, F_n$ subsets using the correspondance \eqref{e:codagephi}.
\item Generate a $N$-sized sample of $(Y, Y^{F_1}, \ldots, Y^{F_n})$.
\item Use this sample in \eqref{e:defestsf}, for $F=F_1, \ldots, F_n$, to obtain the $E$ vector \eqref{e:defE}.
\item Solve problem \eqref{e:defSest} to obtain an $\widehat S$ which estimates $S$.
\end{enumerate}
\hrule

\bigskip

\pagebreak[2]
\noindent
We now consider a different sampling procedure for $\Phi$, which will make it possible to improve on the constants in \eqref{e:FactureBernoulli1} as it will be stated in \eqref{e:FactureRademacher}. Specifically, we sample $\Phi$ using a symmetric Rademacher distribution:
\begin{equation}
\label{e:loirade}
(\Phi_{ji})_{j,i} \text{ are independent}: \; P(\Phi_{ji} = 1)=P(\Phi_{ji}=-1)=1/2. 
\end{equation}
\noindent
The following theorem is the equivalent of Theorem \ref{t:linfini} for Rademacher designs.
\begin{thm}[$\ell^\infty$ error bound] 
\label{t:linfinirade}
Suppose that:
\begin{enumerate}
\item $\epsilon$ is a centered Gaussian vector whose covariance matrix has $\sigma^2$ as largest eigenvalue;
\item $\delta=\frac{1}{7 \delta' s}$ for some real $\delta'>1$;
\item $r=A\sigma\sqrt{\frac{\ln p}{n}}$ for some $A>2\sqrt 2$. 
\end{enumerate}
Let:
\begin{eqnarray*}
t &=& \frac 3 2 \left( 1 + \frac{ 16 }{ 5 (\delta'-1) } \right) r \\
\alpha&=&  1- \left(1-p^{1-A^2/8}\right)\left(1-\exp \left( - n \frac{49 \delta^2 s^2}{2}  + 2 \ln p \right)\right).
\end{eqnarray*}
Then, with probability at least $1-\alpha$, any solution $\widehat S$ of \eqref{e:defSest} satisfies:
\[
 \max_{i=1,\ldots,p} | \widehat S_i - S_i | \leq t.
\]
\end{thm}

\begin{proof}
A proof can be found in Appendix \ref{proof:TheoRade}.
\end{proof}

\noindent
\textbf{Remark 1: } For the probability above to be greater than zero, it is necessary to have:
\begin{equation}
\label{e:FactureRademacher}
n \geq C s^2 \ln p
\end{equation}
for some constant $C>0$.

\noindent
\textbf{Remark 2: } Support recovery property (Corollary \ref{c:supprade}) also holds in this context.

\subsection{Numerical experiments}
\label{sec:NumExp}
\subsubsection{LASSO convergence paths}
In this section, we perform a numerical test of the "Randomized Pick-Freeze" estimation procedure for Bernoulli and Rademacher matrices, summarized respectively on pages \pageref{summaryber} and \pageref{summaryrad}. We use the following model:
\[ Y = f(X_1, \ldots, X_{300}) = X_1^2 + 4 X_1 + 4 X_2 + 10 X_3, \]
hence $p=300$ and $s=3$, with $(X_i)_{i=1,\ldots,120}$ iid uniform on $[0,1]$. It is easy to see that, in this model, we have $S_3 > S_1 > S_2 > 0$ and $S_i=0$ for all $i>3$. The tests are performed by using $n=30$. The obtained LASSO regularization paths (ie., the estimated indices for different choices of the penalization parameter $r$)  are plotted in Figures \ref{f:1} (for Bernoulli design matrix with parameter $\mu=1/2$) and \ref{f:2} (for Rademacher design matrix). The Monte-Carlo sample size used are $N=3000$ and $N=2000$, respectively for Bernoulli and Rademacher designs. This difference in sample sizes accounts for the increase in the number of required evaluations of the $f$ function when a Rademacher design is used (as, in this case, each replication is a difference of two pick-freeze estimators on the same design). 

We observe that the Rademacher design seems to perform better (as LASSO convergence is faster) than the Bernoulli design, in accordance with the remarks made in the beginning of Section \ref{sec:RademacherBreaking}. Both designs perfectly recover the active variables (the support of $S$), as well as the ordering of indices. Note that the proposed method requires only $30 \times 2 \times 3000=180 000$  evaluations of the $f$ function to estimate the 300 Sobol indices, while a classic one-by-one pick-freeze estimation with the same Monte-Carlo sample size would require $3000 \times (300+1) = 903 000$ evaluations of $f$.

\begin{figure}
\begin{center}
\includegraphics[width=7cm,angle=-90]{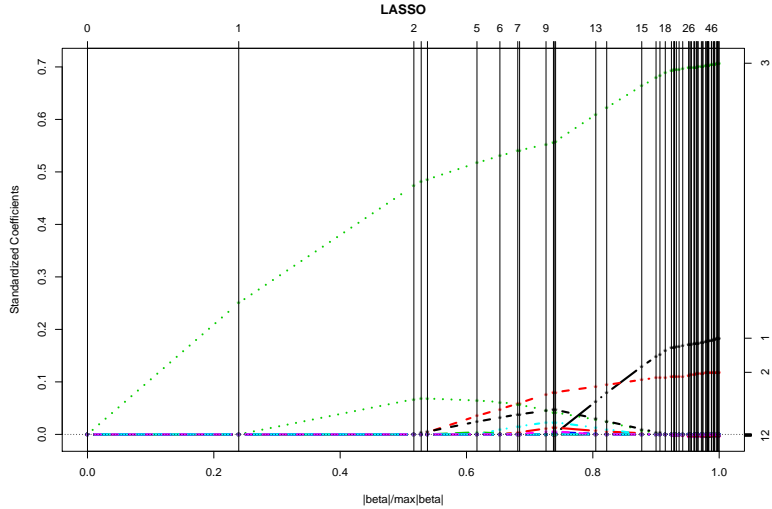}
\caption{LASSO convergence path for a Bernoulli design.}
\label{f:1} 
\end{center}
\end{figure}

\begin{figure}
\begin{center}
\includegraphics[width=7cm,angle=-90]{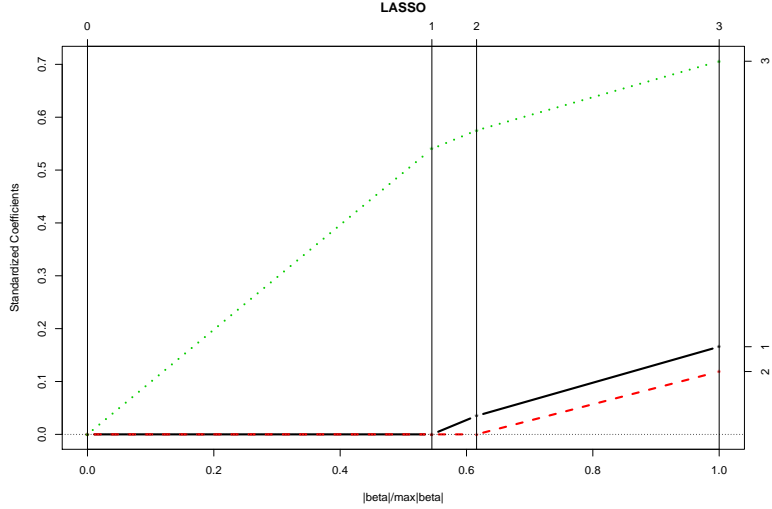}
\caption{LASSO convergence path for a Rademacher design.}
\label{f:2} 
\end{center}
\end{figure}

\subsubsection{Illustration of $\ell^\infty$ error bounds}
We now present a synthetic example which shows the performance of the Rademacher RPF algorithm, used with Theorem \ref{t:linfinirade} and the support recovery corollary.

Suppose that we work on a model with $p=30000$ inputs, with only $s=3$ of them have a nonzero Sobol index. We postulate that all the $\widehat S_i$ estimators, as well as the $\widehat S_F^\Delta$ have standard Gaussian distribution.  By using $N=10^6$ and $n=100$ in Theorem \ref{t:linfinirade}, we get that the $t$ error bound given in this theorem is $t=0.03$, with probability greater than $1-\alpha=95\%$. Hence, by doing calling $3 N n = 3 \times 10^8$ to the $f$ function, one can correctly identify parameters whose Sobol indices are greater than $0.03$.

On the other hand, when using classical one-by-one Sobol index estimation, one has to perform $p=30000$ independent estimations of Sobol index confidence intervals, at level $1 - 0.95^{1/30000}=1.71 \times 10^{-6}$ (by using Šidák correction). 
 From the quantiles of the Gaussian distribution, the length of these intervals is $9.568/\sqrt N$. Hence, getting confidence intervals of width $0.03$ require $N'=(9.568/0.03)^2 \approx 101720$ sample size. Hence, the total cost for this method is $2 N' (p+1)=6 103 200 000\approx 61 \times 10^8$ calls to the $f$ function.

\section{Breaking the square-root bottleneck}
\label{sec:RademacherBreaking}

In the beginning of this paper, we have showed results that are limited by the constraint $n\geq C s^2 \log p$ for some constant $C$. This limitation is due to the use of the mutual incoherence property in the proofs, which is heuristically bounded by Welch's bound \cite{welch1974lower}. We now present a new approach, based on Universal Distortion Property \cite{de2012remark} and a relaxed version of the coherence (see Lemma \ref{lem:UDPimpliesSupportRecovery} in Appendix \ref{sec:GeneralResults}) which enables to break this ``bottleneck'' for Rademacher designs. Note that applying this approach for Bernoulli designs leads to a new proof of the above stated results. For sake of completeness, we give these proofs in Appendix \ref{App:AppendixGraph}. This appendix covers the frame of exact support recovery using Thresholded-Lasso using adjacency matrix as design. In this section we focus on Rademacher designs defined by \eqref{e:loirade}, namely $(\Phi_{ji})_{j,i}$ are independent and for all $i,j$, $\mathbb P(\Phi_{ji} = \pm 1)=1/2$. 

\begin{thm}[Exact recovery with Rademacher designs]
\label{thm:tiebreak}
There exists universal constants $C_{1},C_{2},C_{3}>0$ such that the following holds. Let $c>1$ and $\Phi\in\{\pm1\}^{n\times p}$ a Rademacher matrix drawn according to \eqref{e:loirade} with:
\begin{itemize}
\item $n\geq  n_{0}:=C_{1}s\log(C_{2} p)$,
\item $s\geq6(2+c)/C_{1}$,
\item $\epsilon\sim\mathcal N(0,\Sigma_{n})$ and the covariance diagonal entries enjoy $\Sigma_{i,i}\leq \sigma^{2}$.
\end{itemize}
 Let $\hat S$ be any solution to \eqref{e:defSest} with regularizing parameter:
\[
r\geq r_{1}:=45\,\sigma\Big[\frac{c\log p}{n}\Big]^{1/2}\,,
\]
Then, with a probability greater than $1-3p^{-c}-2\exp(-C_{3}n)$,
\begin{equation}
\label{eq:RademacherLinftyConsitency}
\lVert\hat S-S\lVert_{\infty}\leq\sigma\sqrt{\frac{n_{0}}{n}}\Big[\frac{r}{r_{1}}\Big]\Big[C'_{1}+\frac{C'_{2}}{\sqrt s}\Big]\sqrt s\,,
\end{equation}
where $C'_{1}=35869(c({2+c}))^{1/2}/C_{1}$ and $C'_{2}=46.31c^{1/2}/C_{1}^{1/2}$.
\end{thm}

\begin{proof}
\label{proof:RadeBreak}
\noindent
$\bullet$ Invoke Lemma \ref{lem:RademacherCorrelationPredictors} to get that:
\[
\max_{1\leq k\neq l\leq p}\frac1n\lvert\sum_{j=1}^{n}\Phi_{j,k}\Phi_{j,l}\lvert\ \leq \Big[\frac{(2+c)8}{3C_{1}}\Big]^{1/2}\frac1{\sqrt s}\,,
\]
with probability greater than $1-2p^{-c}$.

$\bullet$ Set $r_{0}:=\sigma(2c\log p/n)^{1/2}$ and $Z_{i}=(1/n)\Phi^{\top}\epsilon$. Observe that $Z_{i}$ is centered Gaussian random variable with variance less than $\sigma^{2}/n$. Taking union bounds, it holds:
\begin{align*}
\mathbb P[(1/n)\lVert\Phi^{\top}\epsilon\lVert_{\infty}> r_{0}]\leq \sum_{i=1}^{p}\mathbb P[\lvert Z_{i}\lvert> \sqrt{2c}\sqrt{\log p}\,\sigma/\sqrt n]\leq p^{1-c}\,,
\end{align*}
using $\lVert\Phi_{i}\lVert_{2}^{2}= n$ and the fact that, for $\sqrt{2c}\sqrt{\log p}\geq\sqrt{2\log 2} $, we have:
\[
\mathbb P[\lvert\mathcal N(0,1)\lvert>\sqrt{2c}\sqrt{\log p}]\leq\frac1{\sqrt{\pi\log 2}}\exp(-c\log p)\leq p^{-c}\,.
\]
$\bullet$ From Lemma \ref{lem:RademacherRIP} and Lemma \ref{lem:RademacherUDP} with $\delta=9/50$ and $\kappa=4/9$, it holds that, with a probability greater than $1-2\exp(-C_{3}n)$, for all $\gamma\in\R^{p}$ and for $T\subseteq\{1,\dotsc,p\}$ such that $\abs T\leq s$,
\begin{equation*}
\norm{\gamma_{T}}_{1}\leq 4.4128\Big(\frac sn\Big)^{1/2}\norm{\Phi\gamma}_{2}+\frac49\norm\gamma_{1}\,.
\end{equation*}
Observe that $C_{1}=5/(c_{1}\delta^{2})$, $C_{2}=c_{2}/\delta^{2}$ and $C_{3}=c_{3}C_{1}$ where $c_{1}, c_{2}, c_{3}$ are universal constants appearing in Lemma \ref{lem:RademacherRIP} and $\delta=9/50$.

\noindent
$\bullet$ Invoke Lemma \ref{lem:UDPimpliesSupportRecovery} with parameters $\rho=4.4128/\sqrt n$, $\kappa=4/9$, $\theta_{2}=1$ and $\theta_{1}=((2+c)8/(3C_{1}s))^{1/2}$, to get that for all regularizing parameter $r\geq r_{1}:= 31.74r_{0}$,
\begin{equation*}
\lVert\hat S-S\lVert_{\infty}\leq\Big(1.0316+799\Big(\frac{2+c}{C_{1}}\Big)^{1/2}\sqrt s\Big)r\,,
\end{equation*}
on the event $\{(1/n)\lVert{\Phi^{\top}\epsilon}\lVert_{\infty}\leq r_{0}\}$.
\end{proof}

\begin{rem}
Observe that \eqref{eq:RademacherLinftyConsitency} reads:
\[
\lVert\hat S-S\lVert_{\infty}\leq [\frac{r}{r_{1}}][C'_{1}+\frac{C'_{2}}{\sqrt s}]\,\sigma\,\sqrt{\frac{C_{1}s^{2}\log(C_{2} p)}n}\,.
\]
where $C_{1},C_{2},C'_{1},C'_{2}>0$ are constants. It shows that, for all $\alpha>0$, Thresholded-lasso exactly recovers the true support if the non-zero coefficients are above a threshold that is proportional to $\sigma s^{\frac{1-\alpha}2}$ from $n=\mathcal O(s^{1+\alpha}\log p)$ observations. Hence, we have tackled the regime $0<\alpha<1$ where the elementary analysis of Theorem \ref{t:linfinirade} fails to  be applicable.
\end{rem}


\section{Conclusions}

We have presented a new and performant method for estimating Sobol indices in high-dimensional additive models. We have shown that this method can  lead to very good results in terms of computational costs. Besides, the error analysis of our algorithm led us to propose the results in Section \ref{sec:RademacherBreaking}, which are also of interest outside of the Sobol indices context, and which gives support recovery property for thresholded LASSO that are, to our best knowledge, greatly improving the results of the literature.


\appendix

\section{Proof of the theorems}
\label{App:ProofA}
\subsection{Proof of Theorem \ref{t:linfini}}
\label{proof:TheoBer}
We capitalize on \cite{lounici2008sup, buhlmann2011statistics,zhao2006model} to prove sup-norm error bound when the design matrix has Bernoulli distribution. 

\subsubsection*{Step 1: Rescaling}
We rewrite \eqref{e:modlin} as $\tilde E = \tilde \Phi S + \tilde \epsilon$ where:
\[ \tilde E = \frac{1}{\sqrt\mu} E, \;\; \tilde \Phi = \frac{1}{\sqrt\mu} \Phi, \;\; \tilde \epsilon = \frac{1}{\sqrt\mu}\epsilon. \]
Note $\widehat S$ satisfies:
\[
\widehat S \in \argmin_{U \in \R^p} \,
\left( \frac{1}{n} \norm{ \tilde E - \tilde \Phi U }_2^2 + 2 \tilde r \norm{U}_1 \right)
\]
with 
\begin{equation}
\label{e:defrtilde}
\tilde r=r/\mu=A\sigma\sqrt{\frac{1+\delta}{\mu}}\sqrt{\frac{\ln p}{n}}.
\end{equation}

\subsubsection*{Step 2: Expectation and concentration}
We define:
\[ \Psi = \frac 1n \tilde\Phi^T \tilde\Phi = \frac{1}{n \mu} \Phi^T \Phi. \]
Thanks to the rescaling above, we have, for all $i=1, \ldots,p$:
\[ \E( \Psi_{ii} ) = \frac 1n \sum_{k=1}^n \E( \tilde \Phi_{ki}^2 ) = 1, \]
and, for all $j=1, \ldots, p$, $j \neq i$:
\[ \E( \Psi_{ij} ) = \frac 1n \sum_{k=1}^n \E( \tilde \Phi_{ki} \tilde \Phi_{kj} ) = \mu. \]
Besides, Hoeffding's inequality \cite{Hoeffding1963} gives that for all $i=1,\ldots,p$ and any $\delta>0$,
\[ \mathbb P \left(|\Psi_{ii}-1|\geq\delta \right) = \mathbb P\left(\left|\frac 1n \sum_{k=1}^n (\Phi_{ki}^2 - \mu)\right|\geq\delta\mu\right) \leq 2 \exp( - 2 n \delta^2 \mu^2 ), \] 
and, similarly, for any $j\neq i$,
\[ \mathbb P \left(|\Psi_{ij}-\mu|\geq\delta \right) \leq 2 \exp \left(- 2 n \delta^2 \mu^2 \right). \] 
Thus, by union bound:
\[ \mathbb P \left(\max_{i=1,\ldots,p} |\Psi_{ii}-1|\geq\delta \right) \leq 2 \exp \left( - 2 n \delta^2 \mu^2 + \ln p \right), \] 
and:
\begin{align*}
\mathbb P \left(\max_{\substack{i=1,\ldots,p \\ j=1,\ldots,p \\ j\neq i}} |\Psi_{ij}-\mu|\geq\delta \right) &\leq 2 \exp \left( - 2 n \delta^2 \mu^2 + \ln \frac{p(p-1)}{2} \right)\,, \\
&\leq \exp \left( - 2 n \delta^2 \mu^2 + 2 \ln p \right).
\end{align*}

\subsubsection*{Step 3: Noise control}
We proceed as in the proof of Lemma 1 of \cite{lounici2008sup}. We define, for $i=1,\ldots,p$:
\[ Z_i = \frac{1}{n} \sum_{j=1}^n \tilde\Phi_{ji} \tilde\epsilon_j = \frac 1n \left( \tilde\Phi^T \tilde\epsilon \right)_i. \]
We define the following event:
\[ \mathcal B = \left\{ \max_{i=1,\ldots,p} |\Psi_{ii}-1|\leq\delta \right\}. 
					 \]
\noindent
For a given $\Phi$, we denote by $\Sigma=\Sigma(\Phi)$ the covariance matrix of $\epsilon$, hence the covariance matrix of $\tilde\epsilon$ is $\Sigma/\mu$. Note that, as a function of $\Phi$, $\Sigma$ is also a random variable. We also denote by $\Var Z_i$ the variance of $Z_i$ for a fixed $\Phi$, which is also a $\Phi$-mesurable random variable. Conditionally on $\mathcal B$, we have:
\begin{eqnarray*}
\Var Z_i &=& \frac{1}{n^2} \Var\left[ \left(\tilde\Phi^T \tilde\epsilon\right)_i \right] \\
&=& \frac{1}{\mu n^2} e_i^T \left( \tilde\Phi^T \Sigma \tilde\Phi \right) e_i \text{ where } (e_i)_k=\left\{\begin{array}{l} 1 \text{ if } i=k \\ 0 \text{ else} \end{array}\right. \\
&=& \frac{1}{\mu n^2} (\tilde \Phi e_i)^T \Sigma (\tilde \Phi e_i) \\
&\leq& \frac{1}{\mu n^2} \sigma^2 \norm{\tilde\Phi e_i}_2^2 \\
&=& \frac{1}{n \mu} \sigma^2 e_i^T \Psi e_i \\
&\leq& \frac{1}{n \mu} \sigma^2 (1+\delta) \text{ as } \mathcal B \text{ holds. }
\end{eqnarray*}
Now consider the following event:
\[ \mathcal A = \bigcap_{i=1}^p \{ |Z_i|\leq \frac{\tilde r}{2} \}. \]
We have:
\[ 
\mathbb P(\mathcal A \cap \mathcal B) = \mathbb P(\mathcal A | \mathcal B) \mathbb P(\mathcal B). \]
From union bound and standard results on Gaussian tails, we get:
\begin{eqnarray*}
\mathbb P(\mathcal A | \mathcal B) &\geq& 1 - p \exp\left( - \frac{n \mu}{2 \sigma^2 (1+\delta)} \left(\frac{\tilde r}{2}\right)^2 \right) \\
&\geq& 1 - p^{1-\frac{A^2}{8}}\,
\end{eqnarray*}
by using \eqref{e:defrtilde}. Hence, step 2 gives:
\[ \mathbb P(\mathcal A \cap \mathcal B) \geq \left( 1 - p^{1-A^2/8} \right) \left( 1 -  2 \exp \left( - 2 n \delta^2 \mu^2 + \ln p \right) \right). \] 

\noindent
\emph{Remark: } following Remark 3 (given after the statement of the proven theorem), one can precisely account for the non-gaussianity of the $\epsilon$ noise by subtracting a correction term to minor $\mathbb P(\mathcal A|\mathcal B)$, by using the Berry-Esseen theorem for the $\hat S$ estimator given in \cite{gamboa2013statistical}.

\medskip

\noindent
Now suppose that $\mathcal A \cap \mathcal B$ is realized. We have:
\[ \frac 1n \norm{ \tilde\Phi^T \epsilon }_\infty \leq \frac{\tilde r}{2}, \]
where
\[ \norm{ v }_\infty = \max |v_i|. \]
Set $\Delta=S - \widehat S$. We have:
\begin{eqnarray*}
\norm{\Psi\Delta}_\infty &=& \frac 1 n \norm{ \tilde \Phi^T \tilde \Phi \Delta }_\infty\,, \\
&=& \frac 1 n \norm{ \tilde \Phi^T \tilde \Phi S - \tilde \Phi^T \tilde \Phi \widehat S }_\infty\,, \\
&=& \frac 1 n \norm{ \tilde \Phi^T \tilde E - \tilde\Phi^T \epsilon - \tilde\Phi^T \tilde\Phi \widehat S}_\infty\,, \\
&\leq& \frac 1 n \norm{ \tilde \Phi^T \left( \tilde E - \tilde\Phi\widehat S \right) }_\infty + \frac 1 n \norm{ \tilde \Phi^T \epsilon }_\infty\,.
\end{eqnarray*}
As the Dantzig constraint:
\[ \norm{ \frac 1n \tilde\Phi^T \left(\tilde E - \tilde\Phi \widehat S \right) }_\infty \leq \tilde r \]
holds, see \cite{lounici2008sup}, we have:
\begin{equation}\label{e:step3eq} \norm{ \Psi \Delta }_\infty \leq \frac{3 \tilde r}{2}. \end{equation}

\subsubsection*{Step 4: Control of $\norm{\Delta}_1$}
\paragraph{Step 4a: Majoration of ${\Delta^T \Psi \Delta}$.} We have, on the event $\mathcal A \cap \mathcal B$:
\begin{eqnarray*}
\abs{\Delta^T \Psi \Delta} &\leq& \norm{\Psi\Delta}_\infty \norm{\Delta}_1 \\
&\leq& \frac{3\tilde r}{2} \left( \norm{\Delta_J}_1 + \norm{\Delta_{J^c}}_1 \right),
\end{eqnarray*}
by introducing the $\Delta_J$ and $\Delta_{J^c}$ vectors defined by:
\[ \left(\Delta_J\right)_i = \left\{ \begin{array}{l} \Delta_i \text{ if } S_i \neq 0 \\ 0 \text{ else} \end{array} \right. \;\;\;
\left(\Delta_{J^c}\right)_i = \left\{ \begin{array}{l} 0 \text{ if } S_i \neq 0 \\ \Delta_i \text{ else} \end{array} \right. \]
We recall that $\norm{\Delta_{J^c}}_1 \leq 3 \norm{\Delta_J}_1$ (see \cite{lounici2008sup}, Lemma 1, (9)). Hence, on $\mathcal A \cap \mathcal B$,
\begin{equation}\label{e:majo} \abs{\Delta^T \Psi \Delta} \leq 6 \tilde r \norm{\Delta_J}_1. \end{equation}

\paragraph{Step 4b: Minoration of ${\Delta^T \Psi \Delta}$.} 
Let's introduce the circulant matrix $M$:
\[ M = \begin{pmatrix} 1 & \mu & \cdots & \mu \\
                      \mu&  1  & \ddots & \vdots \\
							 \vdots&\ddots&\ddots&\mu\\
							 \mu & \cdots & \mu & 1 \end{pmatrix} \]
whose smallest eigenvalue is $1-\mu$ (see \cite{graytoeplitz}). Hence:
\begin{eqnarray*}
\Delta^T \Psi \Delta &=& \Delta^T M \Delta + \Delta^T (\Psi - M) \Delta \\
&\geq& (1-\mu) \norm{\Delta}_2^2 - | \Delta^T (\Psi - M) \Delta | \\
&\geq& (1-\mu) \norm{\Delta_J}_2^2 - | \Delta^T (\Psi - M) \Delta | \\
&\geq& \frac{1-\mu}{s} \norm{\Delta_J}_1^2 - | \Delta^T (\Psi - M) \Delta |\,,
\end{eqnarray*}
since $\Delta_J$ has $s$ nonzero components. We have:
\begin{equation}\label{e:majdelpsi} |\Delta^T (\Psi-M) \Delta| \leq \norm{\Delta}_1 \norm{(\Psi-M)\Delta}_{\infty} \leq 4 \norm{\Delta_J}_1 \norm{(\Psi-M)\Delta}_{\infty}. \end{equation}
Now define the event:
\[ \mathcal C = \left\{ \max_{\substack{i=1,\ldots,p \\ j=1,\ldots,p \\ j\neq i}} |\Psi_{ij}-\mu|\geq\delta \right\}. \]
It is clear that, on $\mathcal B \cap \mathcal C$, all entries of $\Psi-M$ are absolutely bounded by $\delta$. Hence, on $\mathcal B \cap \mathcal C$, 
\[ \norm{(\Psi-M)\Delta}_{\infty} \leq \delta \norm{\Delta}_1 \leq 4 \delta \norm{\Delta_J}_1, \]
and, by \eqref{e:majdelpsi}:
\[ |\Delta^T (\Psi-M) \Delta| \leq 16 \delta \norm{\Delta_J}_1^2, \]
which gives:
\begin{equation}\label{e:mino} \Delta^T \Psi \Delta \geq \left( \frac{1-\mu}{s} - 16 \delta \right) \norm{\Delta_J}_1^2. \end{equation}

\paragraph{Step 4c: Majoration of $\norm{\Delta}_1$.} By using \eqref{e:majo} and \eqref{e:mino}, we get that on $\mathcal A \cap \mathcal B \cap \mathcal C$:
\[ \norm{\Delta_J}_1 \leq \frac{ 6 \tilde r }{ \frac{1-\mu}{s} - 16\delta }, \]
hence:
\begin{equation}\label{e:majdelta1} \norm{\Delta}_1 \leq \frac{ 24 \tilde r }{ \frac{1-\mu}{s} - 16\delta }. \end{equation}

\subsubsection*{Step 5: Majoration of $\norm{\Delta}_\infty$} 
On $\mathcal A \cap \mathcal B \cap \mathcal C$, we have:
\begin{eqnarray*}
\norm{\Delta}_\infty &\leq& \norm{\Psi\Delta}_\infty + \norm{\Psi\Delta-\Delta}_\infty \\
&\leq& \frac{3 \tilde r}{2} + \norm{(\Psi-\text{Id})\Delta}_\infty \;\; \text{ by using } \eqref{e:step3eq} \\
&\leq& \frac{3 \tilde r}{2} + (\mu+\delta) \norm{\Delta}_1  \;\; \text{ since each entry in } \Psi-\text{Id} \text{ is less than } \mu+\delta \\
&\leq& \left( \frac 3 2 + \frac{ 24 (\mu+\delta) }{ \frac{1-\mu}{s} - 16\delta } \right) \tilde r \;\; \text{ by using } \eqref{e:majdelta1} 
\end{eqnarray*}
\noindent
To finish, it is easy to see, using step 2, that $\mathbb P(\mathcal A \cap \mathcal B \cap \mathcal C) \geq 1 - \alpha$. \hfill $\square$


\subsection{Proof of Theorem \ref{t:linfinirade}}
\label{proof:TheoRade}
We rely on the result of \cite{lounici2008sup}. Observe that:
\[ \Psi = \frac 1n \Phi^T \Phi. \]
We have for all $j=1, \ldots, p$, $j \neq i$:
\[ \E( \Psi_{ij} ) = \frac 1n \sum_{k=1}^n \E( \tilde \Phi_{ki} \tilde \Phi_{kj} ) = 0. \]
Hence, for any $\delta>0$, Hoeffding's inequality and union bound give:
\[ \mathbb P \left(\max_{\substack{i=1,\ldots,p \\ j=1,\ldots,p \\ j\neq i}} |\Psi_{ij}|\geq\delta \right) \leq \exp \left( - n \frac{\delta^2}{2} + 2 \ln p \right). \]
We also notice that $\Psi_{ii}=1$ for all $i$. Hence, Assumptions 1 and 2 of Theorem 1 in \cite{lounici2008sup} are satisfied with probability as described in the statement of the theorem.

\section{Exact support recovery using Thresholded-Lasso}
\subsection{A new result}
\label{sec:GeneralResults}
\label{App:WithoutWelch}

\noindent
We begin with two lemmas.
\begin{lem}[Lemma A.2 in \cite{de2012remark}]
Let $r>r_{0}>0$ and $\hat S$ a solution to \eqref{e:defSest} with regularizing parameter $r$. Set $\Delta=\hat S-S$. On the event $\{(1/n)\lVert{\Phi^{\top}\epsilon}\lVert_{\infty}\leq r_{0}\}$, it holds that for all $T\subseteq\{1,\dotsc,p\}$ such that $\abs T\leq s$,
\begin{equation}\label{e:IntermediatUDP}
\frac1{2r}\Big[\frac1{2n}\lVert\Phi\Delta\lVert_{2}^{2}+(r-r_{0})\lVert\Delta\lVert_{1}\Big]\leq\lVert\Delta_{T}\lVert_{1}+\lVert S_{T^{c}}\lVert_{1}\,.
\end{equation}
\end{lem}
\begin{proof}
By optimality in \eqref{e:defSest}, we get:
\[
\frac1{2n}\lVert E-\Phi\hat S\lVert_{2}^{2}+r\lVert \hat S\lVert_{1}\leq \frac1{2n}\lVert \epsilon\lVert_{2}^{2}+r\lVert  S\lVert_{1}\,.
\]
It yields,
\[
\frac1{2n}\lVert \Phi\Delta\lVert_{2}^{2}-\frac1n\langle\Phi^{\top}\epsilon,\Delta\rangle+r\lVert \hat S\lVert_{1}\leq r\lVert  S\lVert_{1}\,.
\]
Let $T\subseteq\{1,\dotsc,p\}$ such that $\abs T\leq s$. We assume that $(1/n)\lVert{\Phi^{\top}\epsilon}\lVert_{\infty}\leq r_{0}$. Invoking H\"older's inequality, we have:
\[
\frac1{2n}\lVert \Phi\Delta\lVert_{2}^{2}+r\lVert \hat S_{J^{c}}\lVert_{1}\leq r(\lVert S_{J}\lVert_{1}-\lVert \hat S_{J}\lVert_{1})+r\lVert S_{J^{c}}\lVert_{1}+r_{0}\lVert  \Delta\lVert_{1}\,.
\]
Adding $r\lVert S_{J^{c}}\lVert_{1} $ on both sides, observe that:
\begin{equation}\label{e:tubeIntermediate}
\frac1{2n}\lVert \Phi\Delta\lVert_{2}^{2}+(r-r_{0})\lVert \Delta_{J^{c}}\lVert_{1}\leq (r+r_{0})\lVert \Delta_{J}\lVert_{1}+2r\lVert S_{J^{c}}\lVert_{1}\,.
\end{equation}
Adding $(r-r_{0})\lVert \Delta_{J}\lVert_{1} $ on both sides, we conclude the proof.
\end{proof}

\begin{lem}[Theorem 2.1 in \cite{de2012remark}]\label{lem:ConsistL1}
Assume that for all $\gamma\in\R^{p}$, for all $T\subseteq\{1,\dotsc,p\}$ such that $\abs T\leq s$,
\begin{equation}\label{eq:UDP}
\norm{\gamma_{T}}_{1}\leq\rho\sqrt s\norm{\Phi\gamma}_{2}+\kappa\norm\gamma_{1}\,.
\end{equation}
where $\rho>0$ and $1/2>\kappa>0$. Moreover, assume that the regularizing parameter $r$ of the convex program \eqref{e:defSest} enjoys $r>{r_{0}}/({1-2\kappa})$. Then, on the event $\{(1/n)\lVert{\Phi^{\top}\epsilon}\lVert_{\infty}\leq r_{0}\}$, any solution $\hat S$ to \eqref{e:defSest} satisfies:
\[
\lVert\hat S-S\lVert_{1}\leq\frac{2rn\rho^{2}s}{1-(r_{0}/r)-2\kappa}\,.
\]
\end{lem}
\begin{proof}
Assume that $(1/n)\lVert{\Phi^{\top}\epsilon}\lVert_{\infty}\leq r_{0}$. Using \eqref{eq:UDP} and \eqref{e:IntermediatUDP} with $J=T$, the support of $S$, we get:
\[
\frac1{2r}\Big[\frac1{2n}\lVert\Phi\Delta\lVert_{2}^{2}+(r-r_{0})\lVert\Delta\lVert_{1}\Big]\leq\rho\sqrt s\norm{\Phi\Delta}_{2}+\kappa\norm\Delta_{1}\,,
\]
where $\Delta=\hat S-S$. It yields,
\[
\Big[\frac1{2}(1-\frac{r_{0}}r)-\kappa\Big]\lVert\Delta\lVert_{1}\leq-\frac1{4rn}\lVert\Phi\Delta\lVert_{2}^{2}+\rho\sqrt s\norm{\Phi\Delta}_{2}\leq rn\rho^{2}s\,,
\]
using the fact that the polynomial $x\mapsto-1/(4rn)x^{2}+\rho\sqrt s x$ is not greater than $rn\rho^{2}s$.
\end{proof}

\noindent 
We deduce the following new result on exact support recovery when using Thresholded-Lasso.

\begin{lem}[Exact support recovery with Thresholded-Lasso]
\label{lem:UDPimpliesSupportRecovery}
Assume that for all $\gamma\in\R^{p}$, for all $T\subseteq\{1,\dotsc,p\}$ such that $\abs T\leq s$,
\begin{equation}\notag
\norm{\gamma_{T}}_{1}\leq\rho\sqrt s\norm{\Phi\gamma}_{2}+\kappa\norm\gamma_{1}\,.
\end{equation}
where $\rho>0$ and $1/2>\kappa>0$. Moreover, assume that:
\begin{equation*}
\max_{1\leq k\neq l\leq p}\ \frac1n\lvert\sum_{j=1}^{n}\Phi_{j,k}\Phi_{j,l}\lvert\ \leq\theta_{1}\quad\mathrm{and}\quad\forall i,\ \frac1n\lVert\Phi_{i}\lVert_{2}^{2}\geq\theta_{2}\,,
\end{equation*}
where $\Phi_{i}$ denotes the columns of $\Phi$. Let $r_{0}>0$ and suppose that the regularizing parameter $r$ of the convex program \eqref{e:defSest} enjoys:
\[
r>\frac{r_{0}}{1-2\kappa}\,.
\]
Then, on the event $\{(1/n)\lVert{\Phi^{\top}\epsilon}\lVert_{\infty}\leq r_{0}\}$, any solution $\hat S$ to \eqref{e:defSest} satisfies:
\begin{equation*}
\lVert\hat S-S\lVert_{\infty}\leq\frac1{\theta_{2}}\Big[1+\frac{r_{0}}r+\frac{2n\theta_{1}\rho^{2}s}{1-(r_{0}/r)-2\kappa}\Big]r\,.
\end{equation*}
\end{lem}
\begin{proof}
The first order optimality conditions of the convex program \eqref{e:defSest} shows that there exists $\tau\in\mathbb R^{p}$ such that $\lVert\tau\lVert_{\infty}\leq1$ and:
\[
\frac1n\Phi^{\top}(E-\Phi\hat S)=r\tau\,.
\]
Set $\Delta=\hat S-S$ and $\Psi=(1/n)\,\Phi^{\top}\Phi$. We assume that $(1/n)\lVert{\Phi^{\top}\epsilon}\lVert_{\infty}\leq r_{0}$. It holds:
\begin{equation}\label{e:tube2}
\lVert\Psi\Delta\lVert_{\infty}\leq r+r_{0}\,.
\end{equation}
Moreover, Lemma \ref{lem:ConsistL1} shows that:
\begin{equation}\label{e:ConsistL1}
\lVert\Delta\lVert_{1}\leq \frac{2nr\rho^{2}s}{1-r_{0}/r-2\kappa}\,.
\end{equation}
Since each entry in the matrix $\Psi-\mathrm{Diag}(\lVert\Phi_{1}\lVert_{2}^{2}/n,\ldots,\lVert\Phi_{p}\lVert_{2}^{2}/n)$ is less than $\theta_{1}$, we deduce that:
\begin{align*}
\theta_{2}\lVert\Delta\lVert_{\infty}&\leq\lVert\Psi\Delta\lVert_{\infty}+\lVert(\Psi-\mathrm{Diag}(\lVert\Phi_{1}\lVert_{2}^{2}/n,\ldots,\lVert\Phi_{p}\lVert_{2}^{2}/n))\Delta\lVert_{\infty}\,,\\
&\leq r+r_{0}+\theta_{1}\lVert\Delta\lVert_{1}\,,\\
&\leq \Big[1+\frac{r_{0}}r+\frac{2n\theta_{1}\rho^{2}s}{1-r_{0}/r-2\kappa}\Big]r\,,
\end{align*}
using \eqref{e:tube2} and \eqref{e:ConsistL1}.
\end{proof}

\subsection{Expander graphs}
\label{App:AppendixGraph}
This subsection is devoted to a new proof of support recovery of Thresholded-Lasso when using adjacency matrices. Given the binary constraint, we choose $\Phi$ as the adjacency matrix of a bi-partite simple graph $G=(A,B,E)$ where $A=\{1,\ldots,p\}$, $B=\{1,\ldots,n\}$ and $E\subseteq A\times B$ denotes the set of edges between $A$ and $B$. In this model, $(\Phi_{ji})_{j,i}$ is equal to $1$ if there exists an edge between $j\in B$ and $i\in A$, and $0$ otherwise. Assume that $G$ is left regular with degree $d$, i.e. $\Phi$ has exactly $d$ ones per column. Consider unbalanced expander graphs defined as follows.
\begin{defn}[$(s,e)$-unbalanced expander]\label{d:Expander}
A $(s,e)$-unbalanced expander is a bi-partite simple graph $G=(A,B,E)$ with left degree $d$ such that for any $I\subset A$ with $\# I\leq s$, the set of neighbors $N(I)$ of $I$ has size:
\begin{equation}\label{eq:Expansion}
 \#{N(I)}\geq(1-e)\,d\,\# I\,.
\end{equation}
The parameter $e$ is called the {expansion constant}.
\end{defn}

\noindent
We recall that expander graphs satisfy the UDP property, see the following lemma.
\begin{lem}\label{lem:ExpanderUDP}
Let $\Phi\in\R^{n\times p}$ be the adjacency matrix of a $(2s,e)$-unbalanced expander with an expansion constant $e<1/2$ and left degree $d$. If the quantities $1/e$ and $d$ are smaller than $p$ then $\Phi$ satisfies for all $\gamma\in\R^{p}$ and for all $T\subseteq\{1,\dotsc,p\}$ such that $\abs T\leq s$,
\begin{equation*}
\norm{\gamma_{T}}_{1}\leq\frac{\sqrt{s}}{(1-2e)\sqrt d}\norm{\Phi\gamma}_{2}+\frac{2e}{1-2e}\norm\gamma_{1}\,.
\end{equation*}
\end{lem}
\begin{proof}
For sake of completeness, we present the proof given in \cite{deoptimal}. Without loss of generality, we can assume that $T$ consists of the largest, in magnitude, coefficients of $\gamma$. We partition the coordinates into sets $T_{0}$, $T_{1}$, $T_{2}$, ... ,$T_{q}$, such that the coordinates in the set $T_{l}$ are not larger than the coordinates in $T_{l-1}$, $l\geq1$, and all sets but the last one $T_{q}$ have size $s$. Observe that we can choose $T_{0}=T$. Let $\Phi'$ be a sub matrix of $\Phi$ containing rows from $N(T)$, the set of neighbors of $T$. Using Cauchy-Schwartz inequality, it holds
\begin{equation*}
\sqrt{sd}\norm{\Phi\gamma}_{2}\geq \sqrt{sd}\norm{\Phi'\gamma}_{2}\geq \frac{\sqrt{sd}}{\sqrt{\lvert{N(T)}\lvert}}\norm{\Phi'\gamma}_{1}\geq \norm{\Phi'\gamma}_{1}.
\end{equation*}
From \cite{berinde2008combining}, we know that:
\begin{equation}\label{eq:RIP1}
\norm{\Phi\gamma_{T}}_{1}\geq d(1-2e)\norm{\gamma_{T}}_{1}\,,
\end{equation}
 Moreover, Eq. \eqref{eq:RIP1} gives:
\begin{align*}
\sqrt{sd}\norm{\Phi\gamma}_{2} \geq&\,\norm{\Phi'\gamma}_{1},\\
\geq &\,\norm{\Phi'\gamma_{T}}_{1}-\sum_{l\geq1}\sum_{(i,j)\in E, i\in T_{l},j \in N(T)}\abs{\gamma_{i}},\\
\geq &\, d(1-2e)\norm{\gamma_{T}}_{1}-\sum_{l\geq1}\lvert E\cap(T_{l}\times N(T))\lvert\,\min_{i\in T_{l-1}}\abs{\gamma_{i}},\\
\geq &\,  d(1-2e)\norm{\gamma_{T}}_{1}-\frac{1}{s}\sum_{l\geq1}\lvert E\cap(T_{l}\times N(T))\lvert\,\norm{\gamma_{T_{l-{1}}}}_{1}.
\end{align*}
From the expansion property \eqref{eq:Expansion}, it follows that, for $l\geq1$, we have:
 \[
 \lvert N(T\cup T_{l})\lvert\geq d(1-e)\lvert T\cup T_{l}\lvert\,.
 \]
 Hence at most $de 2s$ edges can cross from $T_{l}$ to $N(T)$, and so:
\begin{align*}
\sqrt{sd}\norm{\Phi\gamma}_{2}& \geq d(1-2e)\norm{\gamma_{T}}_{1}-de2\sum_{l\geq1}\norm{\gamma_{T_{l-{1}}}}_{1}/s,\\
& \geq d(1-2e)\norm{\gamma_{T}}_{1}-2de\norm\gamma_{1}.
\end{align*}
\end{proof}

\noindent
Observe the columns $\Phi_{i}$ of the adjacency matrix $\Phi$ have small $\ell_{2}$-norm compared to the $\ell_{2}$-norm of the noise, namely:
\[
\lVert \Phi_{i}\lVert_{2}^{2}= d\ll \sigma^{2} n=\mathbb E(\lVert \epsilon\lVert_{2}^{2})\,.
\]
A standard hypothesis in the exact recovery frame \cite{bickel2009simultaneous,buhlmann2011statistics} is that the signal-to-noise ratio is close to one. This hypothesis is often presented as the empirical covariance matrix has diagonal entries equal to $1$. However, in our setting, the signal-to-noise ratio goes to zero and eventually we observe only noise. To prevent this issue, we use a noise model adapted to the case of sparse designs. Hence, we assume subsequently that the noise level is comparable to the signal power:
\begin{equation}\label{e:NoiseSparseGraphs}
\forall i\in\{1,\ldots,n\}\,,\quad \epsilon_{i}\mathrm{\ is\ Gaussian\ and\ }\Var(\epsilon_{i})\leq\tilde\sigma^{2}\frac{\lVert \Phi_{i}\lVert_{2}^{2}}n\,,
\end{equation}
so that $\lVert \Phi_{i}\lVert_{2}^{2}/\mathbb E(\lVert \epsilon\lVert_{2}^{2})\geq1/\tilde\sigma^{2}$.

\begin{thm}[Exact recovery with expander graphs]\label{t:Expander}
Let $A>\sqrt 2$ and $\Phi\in\mathbb \{0,1\}^{n\times p}$ be the adjacency matrix of a $(2s,e)$-expander graph with expansion constant $1/p<e<1/6$ and left degree $d$. Assume that \eqref{e:NoiseSparseGraphs} holds. Let $\hat S$ be any solution to \eqref{e:defSest} with regularizing parameter:
\[
r\geq r_{1}:=2A\tilde\sigma\Big[\frac{1-2e}{1-6e}\Big]\Big[\frac{d(\log p)^{1/2}}{n^{3/2}}\Big]\,,
\]
Then, with probability greater than $1-p^{1-A^{2}/2}$, it holds:
\begin{equation*}
\lVert\hat S-S\lVert_{\infty}\leq A\tilde\sigma\Big[\frac{\log p}{n}\Big]^{1/2}\Big[1+\frac{2(1-2e)}{1-6e}+\frac{16es}{(1-6e)^{2}}\Big]\frac r{r_{1}}\,.
\end{equation*}
\end{thm}

\begin{proof}
Lemma \ref{lem:ExpanderUDP} shows that for all $\gamma\in\R^{p}$ and for all $T\subseteq\{1,\dotsc,p\}$ such that $\abs T\leq s$,
\begin{equation*}
\norm{\gamma_{T}}_{1}\leq\frac{\sqrt{s}}{ \sqrt d(1-2e)}\norm{\Phi\gamma}_{2}+\frac{2e}{1-2e}\norm\gamma_{1}\,.
\end{equation*}
Moreover, the expansion property implies:
\begin{equation*}
\max_{1\leq k\neq l\leq p}\frac1n\lvert\sum_{j=1}^{n}\Phi_{j,k}\Phi_{j,l}\lvert\leq\frac{2de}n\,.
\end{equation*}
Lemma \ref{lem:UDPimpliesSupportRecovery} with $1/\rho=(1-2e)\sqrt d$, $\kappa=2e/(1-2e)$, $\theta_{1}=2de/n$ and $\theta_{2}=d/n$, shows that for all regularizing parameter $r\geq2r_{0}({1-2e})/({1-6e})$,
\begin{equation*}
\lVert\hat S-S\lVert_{\infty}\leq\frac n{d}\Big[1+\frac{1-6e}{2(1-2e)}+\frac{8es}{(1-2e)(1-6e)}\Big]r\,,
\end{equation*}
on the event $\{(1/n)\lVert{\Phi^{\top}\epsilon}\lVert_{\infty}\leq r_{0}\}$. Finally, set $r_{0}=A\tilde\sigma {d(\log p)^{1/2}}/{n^{3/2}}$ and $Z_{i}=(1/n)\Phi^{\top}\epsilon$. Observe that $Z_{i}$ is centered Gaussian random variable with variance less than $\tilde\sigma^{2}\lVert\Phi_{i}\lVert_{2}^{4}/n^{3}$. Taking union bounds, it holds:
\begin{align*}
\mathbb P[(1/n)\lVert\Phi^{\top}\epsilon\lVert_{\infty}> r_{0}]&\leq\mathbb P[(1/n)\lVert\Phi^{\top}\epsilon\lVert_{\infty}> A\tilde\sigma {d(\log p)^{1/2}}/{n^{3/2}}]\,,\\
&\leq \sum_{i=1}^{p}\mathbb P[\lvert Z_{i}\lvert> A\tilde\sigma {d(\log p)^{1/2}}/{n^{3/2}}]\,,\\
&= \sum_{i=1}^{p}\mathbb P[\lvert Z_{i}\lvert> (\tilde\sigma\lVert\Phi_{i}\lVert_{2}^{2}/n^{3/2})\,A\sqrt{\log p}]\,,\\
&\leq p^{1-A^{2}/2}\,,
\end{align*}
using $\lVert\Phi_{i}\lVert_{2}^{2}=d$ and the fact that, for $A\sqrt{\log p}\geq\sqrt{2\log 2} $, we have:
\[
\mathbb P[\lvert\mathcal N(0,1)\lvert>A\sqrt{\log p}]\leq\frac1{\sqrt{\pi\log 2}}\exp(-c\log p)\leq p^{-A^{2}/2}\,.
\]
\end{proof}
\noindent
Note that, with high probability, a random bi-partite simple graph is a $(s,e)$-unbalanced expander. As a matter of fact, we have the following result using Chernoff bounds and Hoeffding's inequality, see \cite{xu2007further} for instance.

\begin{prop}\label{p:RandomExpander}
Consider $e>0$, $c>1$ and $p\geq 2s$. Then, with probability greater than $1-s\exp(-c\log p)$, a uniformly chosen bi-partite simple graph $G=(A,B,E)$ with $\abs A=p$, left degree $d$ such that:
\begin{equation}
\label{e:BoundsDegreeRandomGraph}
d\leq C_{1}(c,e)\log p\,,
\end{equation}
and number of right side vertices, namely $n=\abs B$, such that:
\begin{equation}
\label{e:BoundsNRandomGraph}
 n\geq C_{2}(c,e)\,s\log p\,,
\end{equation}
where $C_{1}(c,e), C_{2}(c,e), $ do not depend on $s$ but may depend on $e$, is a $(s,e)$-unbalanced expander graph.
\end{prop}
\noindent 
Hence we deduce the following corollary of Theorem \ref{t:Expander}.

\begin{cor}\label{cor:Expander}
Consider $c>1$, $p\geq 4s$ and choose $e=1/12$. Let $\Phi\in\{0,1\}^{n\times p}$ be drawn uniformly according to Proposition \ref{p:RandomExpander} so that $d\leq C_{1}\log p$ and:
\begin{equation}\label{eq:MeasurementExpanders}
n\geq n_{0}:=C_{2}\,s\log p\,,
\end{equation}
with $C_{1},C_{2}$ constants that depend only on $c$. Let $A> [\min({C_{1}}, 2)]^{1/2}$. Let $\hat S$ be any solution to \eqref{e:defSest} with regularizing parameter:
\[
r\geq r_{1}:=3.34A\tilde\sigma\Big[\frac{\log p}{n}\Big]^{3/2}\,,
\]
Then, with probability greater than $1-p^{1-A^{2}/2}-2s\exp(-c\log p)$, it holds:
\begin{equation}\label{e:tresholdRandomSparseGraphs}
\lVert\hat S-S\lVert_{\infty}\leq {51.7\,A\,C_{2}^{-1/2}\,\tilde\sigma}\,\Big[\frac r{r_{1}}\Big]\,\Big[\frac{n_{0}}{n}\Big]^{1/2}\sqrt s\,.
\end{equation}
\end{cor}
\begin{rem} Observe that \eqref{e:tresholdRandomSparseGraphs} is also consistent with the regime $n=\mathcal O(s^{2}\log p)$. In this case, we uncover that $\lVert\hat S-S\lVert_{\infty}\leq (\mathrm{cst})\,\tilde\sigma$. Namely, the thresholded lasso faithfully recovers the support of entries whose magnitudes are above the noise level.
\end{rem}


\subsection{Bernoulli designs}
\label{App:BerDesigns}
We can relax the hypothesis on the left-regularity using a Bernoulli design that mimics the uniform probability on $d$-regular graphs. This model is particularly interesting since one can easily generate a design matrix $\Phi$. 

Recall we consider a Bernoulli distribution with parameter $\mu \in (0,1)$ and $(\Phi_{ji})_{j,i}$ are independently drawn with respect to this distribution, with for all  $i,j$, it holds $\mathbb P(\Phi_{ji} = 1)=\mu=1-\mathbb P(\Phi_{ji}=0)$. We begin with some preliminaries lemmas.

\begin{lem}\label{lem:BernouDegree}
Let $p,n>0$. Let $c>1$. Let $\Phi\in\{0,1\}^{n\times p}$ a Bernoulli matrix drawn according to \eqref{e:loibern} with:
\[
\mu=799(1+c)\frac{\log p}n\,.
\]
If $n\geq799(c+1){\log p}$ then $\Phi$ satisfies for all $i\in\{1,\ldots,p\}$,
\begin{equation}\label{eq:DegreeBernou}
759(1+c)\log p\leq\lVert\Phi_{i}\lVert_{0}=\lVert\Phi_{i}\lVert_{2}^{2}\,\leq828(1+c)\log p\,,
\end{equation}
and
\[
\max_{1\leq k\neq l\leq p}\frac1n\lvert\sum_{j=1}^{n}\Phi_{j,k}\Phi_{j,l}\lvert\ \leq 879(1+c)\frac{\log p}n\,,\]
with a probability greater than $1-(1+2p)p^{-c}$.
\end{lem}
\begin{proof}
Let $i\in\{1,\ldots,p\}$ and consider $Y_{i}=\Phi_{1,i}+\ldots+\Phi_{n,i}$. Observe that Chernoff bound reads:
\[
\mathbb P(Y_{i}\geq n(\mu+\delta))\leq \exp(-n\,\mathbf H(\mu+\delta\|\mu))
\]
where $\mathbf H(a\|b)$ denotes the Kullback-Leibler divergence between two Bernoulli random variables with parameter $a$ and $b$, namely:
\[
\mathbf H(a\|b)=a\log(a/b)+(1-a)\log((1-a)/(1-b))\,.
\]
Observe that the second derivative of $x\mapsto\mathbf H(\mu+x\|\mu)$ is equal to $1/((\mu+x)(1-\mu-x))$ and is bounded from below by $1/(\mu+\delta)$ on $[\mu,\mu+\delta]$. Therefore, 
\begin{equation}\label{eq:KLMinoration}
\mathbf H(\mu+\delta\|\mu)\geq \frac{\delta^{2}}{2(\mu+\delta)}\,.
\end{equation}
Using union bound, we get that:
\[
\mathbb P[\forall i\,,\ Y_{i}\leq 1.036009\,n\mu]\geq 1-\exp[\log p-0.001252\, n\mu]\geq 1-p^{-({c-1})}\,,
\]
as desired. Similarly, one get that:
\begin{equation}\label{eq:KLMajoration}
\mathbf H(\mu-\delta\|\mu)\leq \frac{\delta^{2}}{2\mu}\,,
\end{equation}
and so:
\[
\mathbb P[\forall i\,,\ Y_{i}\geq 0.05004\, n\mu]\geq 1-\exp[\log p-0.001252\, n\mu]\geq 1-p^{-({c-1})}\,.
\]
The second inequality follows from the same analysis:
\[
\mathbb P[\forall k\neq l\,,\ \frac1n\sum_{j=1}^{n}\Phi_{j,k}\Phi_{j,l}\geq \mu^{2}+0.1\mu)]\leq \exp[\log[\frac{p(p-1)}2]-\frac{n}{200(1+1/(10\mu))}]\,.
\]
Observe that $\log({p(p-1)}/2)\leq2\log p$, $1+1/(10\mu)\leq 1.01/\mu$
and $\mu^{2}+0.1\mu\leq 1.01\mu$. Therefore,
\[
\mathbb P[\forall k\neq l\,,\ \frac1n\sum_{j=1}^{n}\Phi_{j,k}\Phi_{j,l}\leq 1.1\mu)]\geq 1-\exp(-1.5(1+c)\log p)\,.
\]
\end{proof}

\begin{lem}\label{lem:BernouExpander}
Let $p,s>0$ and $c>1$. Let $\Phi\in\{0,1\}^{n\times p}$ a Bernoulli matrix drawn according to \eqref{e:loibern} with $\mu=799(1+c)\log p/n$ and: 
\[
n\geq 6491(1+c)s\log p\,.
\]
Then, with a probability greater than $1-p^{-cs}$, the matrix $\Phi$ satisfies the following vertex expansion property:
\begin{equation}\label{eq:EdgeExpansionBernoulli}
\#\{ \mathrm{Supp}(\Phi1\!\!1_{U})\}\geq (6/7)\, \mathrm d_{max}\#{U} \,,
\end{equation}
where $U$ is a subset of $\{1,\ldots,p\}$ of size $s$, $1\!\!1_{U}\in\R^{p}$ denotes the vector with entry $1$ on $U$ and $0$ elsewhere, and $\mathrm d_{max}=828(1+c)\log p$ is the maximal support size of one column of $\Phi$ as shown in \eqref{eq:DegreeBernou}.
\end{lem}
\begin{proof}
The number of subsets of size $s$ can be upper bounded by $\exp(s\log p)$. Observe that the left hand side of \eqref{eq:EdgeExpansionBernoulli} is a random variable $N_{n}$ with the same law as:
\[
N_{n}\mathop=^{d}\sum_{i=1}^{n}Z_{i}\quad\mathrm{where}\ Z_{i}\mathop\sim^{i.i.d}\mathcal B(\nu)\,,
\]
where the Bernoulli parameter $\nu=1-(1-\mu)^{s}$. Using \eqref{eq:KLMajoration}, we get that:
\[
\mathbb P(N_{n}\leq n(\nu-\delta))\leq\exp(-n\delta^{2}/(2\nu))\,.
\]
Set $\delta:=1-0.8883 s\mu-\exp(-s\mu)$ and observe that it holds $s\mu\leq 0.1231$, $\nu\geq 1-\exp(-s\mu)$, and $\delta\geq s\mu(0.1117-0.5s\mu)\geq0.0501s\mu$. We deduce that:
\[
\mathbb P(N_{n}\leq 0.8883 ns\mu)\leq\exp(-0.001255ns\mu)\leq\exp(-(c+1)s\log p)\,,
\]
using $\delta\geq0.0501s\mu$ and $\nu\leq s\mu$.
\end{proof}
\noindent

\begin{lem}\label{lem:BernouUDP}
Let $p>7$, $s>0$ and $c>1$. Let $\Phi\in\{0,1\}^{n\times p}$ a Bernoulli matrix drawn according to \eqref{e:loibern} with:
\begin{itemize}
\item $\mu=799(1+c)\frac{\log p}n$,
\item $n\geq n_{0}:=12982(1+c)s\log p$.
\end{itemize}
Then, with a probability greater than $1-(1+2p+(1-p^{-c})^{-1})p^{-c}$, the matrix $\Phi$ satisfies for all $\gamma\in\R^{p}$ and for all $T\subseteq\{1,\dotsc,p\}$ such that $\abs T\leq s$,
\begin{equation*}
\norm{\gamma_{T}}_{1}\leq0.0551\Big[\frac{{s}}{{(1+c)\log p}}\Big]^{1/2}\norm{\Phi\gamma}_{2}+0.4529\norm\gamma_{1}\,.
\end{equation*}
\end{lem}
\begin{proof}
From Lemma \ref{lem:BernouExpander}, we get that $\Phi$ is the adjacency matrix of a $(2s,1/7)$-expander graph with left degree $d$ enjoying \eqref{eq:DegreeBernou}, namely $\mathrm d_{min}\leq d\leq \mathrm d_{max}$ with $\mathrm d_{min}=759(1+c)\log p$ and $\mathrm d_{max}=828(1+c)\log p$. Observe that the left degree $d$ may depend on the vertex considered. However, note that the proof of Lemma \ref{lem:ExpanderUDP} can be extended to this case. Following the lines of Lemma 9 in \cite{berinde2008combining}, one can check that:
\begin{equation*}
\norm{\Phi\gamma_{T}}_{1}\geq \mathrm d_{min}(1-2(\mathrm d_{max}/\mathrm d_{min})e)\norm{\gamma_{T}}_{1}\,.
\end{equation*}
Similarly, one can check from the proof of Lemma \ref{lem:ExpanderUDP} that for all $\gamma\in\R^{p}$ and for all $T\subseteq\{1,\dotsc,p\}$ such that $\abs T\leq s$,
\begin{equation*}
\norm{\gamma_{T}}_{1}\leq\frac{\sqrt{s}\sqrt{\mathrm d_{max}}}{ \mathrm d_{min}(1-2(\mathrm d_{max}/\mathrm d_{min})e)}\norm{\Phi\gamma}_{2}+\frac{2\mathrm d_{max}e}{\mathrm d_{min}(1-2(\mathrm d_{max}/\mathrm d_{min})e)}\norm\gamma_{1}\,,
\end{equation*}
where $e=1/7$.
\end{proof}

\noindent 
We deduce the following result for Thresholded-Lasso using Bernoulli design matrices.

\begin{thm}[Exact recovery with Bernoulli designs]
Let $p>7$, $s>0$ and $c>1$. Let $\Phi\in\{0,1\}^{n\times p}$ a Bernoulli matrix drawn according to \eqref{e:loibern} with:
\begin{itemize}
\item $\mu=799(1+c){\log p}/n$,
\item $n\geq 12982(1+c)s\log p$,
\item $\epsilon\sim\mathcal N(0,\Sigma_{n})$ and the covariance diagonal entries enjoy $\Sigma_{i,i}\leq \sigma^{2}$.
\end{itemize}
 Let $\hat S$ be any solution to \eqref{e:defSest} with regularizing parameter:
\[
r\geq r_{1}:=9692\,\sigma\,(1+c)\frac{\log p}n\,,
\]
Then, with a probability greater than $1-3p^{1-c}$,
\begin{equation*}
\lVert\hat S-S\lVert_{\infty}\leq
775.36\,\Big[\frac r{r_{1}}\Big]\,\sigma\, s\,.
\end{equation*}
\end{thm}

\begin{proof}
\noindent
$\bullet$ Invoke Lemma \ref{lem:BernouDegree} to get that for all $i\in\{1,\ldots,p\}$,
\[
759(1+c)\log p\leq\lVert\Phi_{i}\lVert_{0}=\lVert\Phi_{i}\lVert_{2}^{2}\,\leq828(1+c)\log p\,,
\]
and:
\[
\max_{1\leq k\neq l\leq p}\frac1n\lvert\sum_{j=1}^{n}\Phi_{j,k}\Phi_{j,l}\lvert\ \leq 879(1+c)\frac{\log p}n\,.
\]
$\bullet$ Set $r_{0}=6\sigma({46c(1+c)})^{1/2}\,{\log p}/n$ and $Z_{i}=(1/n)\Phi^{\top}\epsilon$. Note $Z_{i}$ is centered Gaussian random variable with variance less than $\sigma^{2}\lVert\Phi_{i}\lVert_{2}^{2}/n^{2}$. Taking union bounds, it holds:
\begin{align*}
\mathbb P[(1/n)\lVert\Phi^{\top}\epsilon\lVert_{\infty}> r_{0}]&\leq\mathbb P[(1/n)\lVert\Phi^{\top}\epsilon\lVert_{\infty}> 6\sigma\sqrt{46c(1+c)}{\log p}/n]\,,\\
&\leq \sum_{i=1}^{p}\mathbb P[\lvert Z_{i}\lvert> 6\sigma\sqrt{46c(1+c)}{\log p}/n]\,,\\
&\leq \sum_{i=1}^{p}\mathbb P[\lvert Z_{i}\lvert> \sqrt{2c\log p}\,\sigma\lVert\Phi_{i}\lVert_{2}/n]\,,\\
&\leq p^{1-c}\,,
\end{align*}
using $\lVert\Phi_{i}\lVert_{2}^{2}\leq828(1+c)\log p$ and the fact that, for $\sqrt{2c\log p}\geq\sqrt{2\log 2} $, we have:
\[
\mathbb P[\lvert\mathcal N(0,1)\lvert>\sqrt{2c\log p}]\leq\frac1{\sqrt{\pi\log 2}}\exp(-c\log p)\leq p^{-c}\,.
\]
$\bullet$ From Lemma \ref{lem:BernouUDP}, it holds that for all $\gamma\in\R^{p}$ and for all $T\subseteq\{1,\dotsc,p\}$ such that $\abs T\leq s$,
\begin{equation*}
\norm{\gamma_{T}}_{1}\leq0.0551\Big[\frac{{s}}{{(1+c)\log p}}\Big]^{1/2}\norm{\Phi\gamma}_{2}+0.4529\norm\gamma_{1}\,.
\end{equation*}
\noindent
$\bullet$ Invoke Lemma \ref{lem:UDPimpliesSupportRecovery} with parameters $\rho=0.0551/\sqrt{(1+c)\log p}$, $\kappa=0.4529$, $\theta_{1}=879(1+c){\log p}/n$ and $\theta_{2}=759(1+c)\log p/n$, to get that for all regularizing parameter $r>10.616\,r_{0}$,
\begin{equation*}
\lVert\hat S-S\lVert_{\infty}\leq\frac 
n{759(1+c)\log p}\Big[1+\frac{r_{0}}r+\frac{5.338\,s}{0.0942-r_{0}/r}\Big]r\,,
\end{equation*}
on the event $\{(1/n)\lVert{\Phi^{\top}\epsilon}\lVert_{\infty}\leq r_{0}\}$. Finally, observe that $r_{1}\geq 238.1r_{0}$.
\end{proof}

\begin{cor}[Exact recovery under constant SNR hypothesis]
Let $p>7$, $s>0$ and $c>1$. Let $\Phi\in\{0,1\}^{n\times p}$ a Bernoulli matrix drawn according to \eqref{e:loibern} with:
\begin{itemize}
\item $\mu=799(1+c){\log p}/n$,
\item $n\geq n_{0}:=12982(1+c)s\log p$,
\item assume that \eqref{e:NoiseSparseGraphs} holds, namely $\epsilon\sim\mathcal N(0,\Sigma_{n})$ and the covariance diagonal entries enjoy $\Sigma_{i,i}\leq 759\,\tilde\sigma^{2}\,(1+c)\log p/n$ with $\tilde\sigma>0$.
\end{itemize}
 Let $\hat S$ be any solution to \eqref{e:defSest} with regularizing parameter:
\[
r\geq 0.1886\,\tilde\sigma\,\Big[\frac{12982(1+c)\log p}n\Big]^{3/2}\,,
\]
Then, with a probability greater than $1-3p^{1-c}$,
\begin{equation*}
\lVert\hat S-S\lVert_{\infty}\leq
195.82\,\tilde\sigma\,\Big[\frac r{r_{1}}\Big]\,\Big[\frac{n_{0}}n\Big]^{1/2} \sqrt s\,.
\end{equation*}
\end{cor}
\begin{proof}
Eq. \eqref{eq:DegreeBernou} shows that:
\[
759(1+c)\log p\leq\lVert\Phi_{i}\lVert_{2}^{2}\,\leq828(1+c)\log p\,,
\]
and so $\Sigma_{i,i}\leq\tilde\sigma^{2}\lVert\Phi_{i}\lVert_{2}^{2}/n$, with high probability.
\end{proof}

\subsection{Rademacher Designs}
The result and the proof given on Page \ref{proof:RadeBreak} rely on the following lemmas.
\label{App:RadeClassics}
\begin{lem}[Rademacher designs satisfy RIP] 
\label{lem:RademacherRIP}
There exists universal constants $c_{1},c_{2},c_{3}$ such that the following holds. Let $\delta\in(0,1)$ and $p,n,s'>0$ such that:
\[
s'=\Big\lfloor\frac{c_{1}\delta^{2}n}{\log(c_{2}p/(\delta^{2}n))}\Big\rfloor\,,
\]
then, with probability at least $1-2\exp(-c_{3}n)$, a matrix $\Phi\in\{\pm1\}^{n\times p}$ drawn according to the Rademacher model \eqref{e:loirade} enjoy the RIP property, namely for all $\gamma\in\R^{p}$ such that $\lVert\gamma\lVert_{0}\leq s'$,
\[
n(1-\delta)^{2}\lVert\gamma\lVert_{2}^{2}\leq\lVert\Phi\gamma\lVert_{2}^{2}\leq n(1+\delta)^{2}\lVert\gamma\lVert_{2}^{2}\,.
\]
\end{lem}
\begin{proof}
Numerous authors have proved this result, see Example 2.6.3 and Theorem 2.6.5 in \cite{chafai2011interactions} for instance.
\end{proof}
\begin{lem}[Rademacher designs satisfy UDP] 
\label{lem:RademacherUDP}
There exists universal constants $c_{1},c_{2},c_{3}$ such that the following holds. Let $\delta\in(0,\sqrt 2-1)$ and $s>0$ such that:
\[
5s\leq s':=\Big\lfloor\frac{c_{1}\delta^{2}n}{\log(c_{2}p/(\delta^{2}n))}\Big\rfloor\,,
\]
then, with probability at least $1-2\exp(-c_{3}n)$, a matrix $\Phi\in\{\pm1\}^{n\times p}$ drawn according to the Rademacher model \eqref{e:loirade} enjoy for all $\gamma\in\R^{p}$ and for all $T\subseteq\{1,\dotsc,p\}$ such that $\abs T\leq s$,
\begin{equation*}
\norm{\gamma_{T}}_{1}\leq\rho\sqrt s\norm{\Phi\gamma}_{2}+\kappa\norm\gamma_{1}\,.
\end{equation*}
where:
\begin{itemize}
\item $1/2>\kappa>\big({1+2(({1-\delta})/({1+\delta}))^{\frac12}}\big)^{-1}$,
\item $\sqrt n\rho=\big(\sqrt{1-\delta}+({\kappa_0\!-\!1})/({2\kappa_0})\sqrt{1+\delta}
\big)^ { -1}$\,.
\end{itemize}

\end{lem}
\begin{proof}
The proof follows from Lemma \ref{lem:RademacherRIP} and Proposition 3.1 in \cite{de2012remark}.
\end{proof}

\begin{lem}
\label{lem:RademacherCorrelationPredictors}
Let $c,C_{1}>0$ and $p,n,s>0$ such that $n\geq C_{1}s\log p$ and $s\geq 3(2+c)/C_{1}$. Then, with probability greater than $1-2p^{-c}$, it holds for all $k\neq l\in\{1,\ldots,p\}$,
\[
\frac1n\lvert\sum_{i=1}^{n}\Phi_{i,k}\Phi_{i,l}\vert\leq \Big[\frac{(2+c)8}{3C_{1}}\Big]^{1/2}\frac1{\sqrt s}\,.
\]
\end{lem}

\begin{proof}
Let $k\neq l\in\{1,\ldots,p\}$. Set $X_{i}=\Phi_{i,k}\Phi_{i,l}$ and observe that $X_{i}$ are independent Rademacher random variables. From Bernstein's inequality, it holds for all $0<t<1$,
\[
\mathbb P\Big(|\frac1n\sum_{i=1}^{n}X_{i}|\geq t\Big)\leq2\exp\Big(-\frac{3n}8t^{2}\Big)\,.
\]
Set $t=({(2+c)8}/(3sC_{1}))^{1/2}$ and observe $\#\{k\neq l\}\leq \exp(2\log p)$.
\end{proof}

\bibliographystyle{plain}
\bibliography{biblio}

\def\cprime{$'$}
\begin{thebibliography}{10}

\bibitem{berinde2008combining}
R.~Berinde, A.~C. Gilbert, P.~Indyk, H.~Karloff, and M.~J. Strauss.
\newblock Combining geometry and combinatorics: A unified approach to sparse
  signal recovery.
\newblock In {\em Communication, Control, and Computing, 2008 46th Annual
  Allerton Conference on}, pages 798--805. IEEE, 2008.

\bibitem{bickel2009simultaneous}
P.~J. Bickel, Y.~Ritov, and A.~B. Tsybakov.
\newblock Simultaneous analysis of lasso and dantzig selector.
\newblock {\em The Annals of Statistics}, 37(4):1705--1732, 2009.

\bibitem{buhlmann2011statistics}
P.~L. B{\"u}hlmann and S.~A. van~de Geer.
\newblock {\em Statistics for high-dimensional data}.
\newblock Springer, 2011.

\bibitem{candes2009near}
E.~J. Cand{\`e}s and Y.~Plan.
\newblock Near-ideal model selection by ℓ1 minimization.
\newblock {\em The Annals of Statistics}, 37(5A):2145--2177, 2009.

\bibitem{candes2006stable}
E.~J. Candes, J.~K. Romberg, and T.~Tao.
\newblock Stable signal recovery from incomplete and inaccurate measurements.
\newblock {\em Communications on pure and applied mathematics},
  59(8):1207--1223, 2006.

\bibitem{chafai2011interactions}
D.~Chafa{\i}, O.~Gu{\'e}don, G.~Lecu{\'e}, and A.~Pajor.
\newblock Interactions between compressed sensing, random matrices, and high
  dimensional geometry.
\newblock {\em to appear in "Panoramas et Synth{\`e}ses" (SMF)}, 2013.

\bibitem{de2012remark}
Y.~de~Castro.
\newblock A remark on the lasso and the dantzig selector.
\newblock {\em Statistics and Probability Letters}, 2012.

\bibitem{deoptimal}
Y.~de~Castro.
\newblock Optimal designs for lasso and dantzig selector using expander codes.
\newblock {\em Arxiv preprint arXiv:1010.2457v5}, 2013.

\bibitem{efron2004least}
B.~Efron, T.~Hastie, I.~Johnstone, and R.~J. Tibshirani.
\newblock Least angle regression.
\newblock {\em The Annals of statistics}, 32(2):407--499, 2004.

\bibitem{fuchs2004sparse}
J-J Fuchs.
\newblock On sparse representations in arbitrary redundant bases.
\newblock {\em Information Theory, IEEE Transactions on}, 50(6):1341--1344,
  2004.

\bibitem{gamboa2013statistical}
F.~Gamboa, A.~Janon, T.~Klein, A.~Lagnoux-Renaudie, and C.~Prieur.
\newblock Statistical inference for sobol pick freeze monte carlo method.
\newblock {\em arXiv preprint arXiv:1303.6447}, 2013.

\bibitem{graytoeplitz}
R.~Gray.
\newblock Toeplitz and circulant matrices: A review.
\newblock {\em \url{http://ee.stanford.edu/\~gray/toeplitz.pdf}}, 2002.

\bibitem{Hoeffding1963}
W.~Hoeffding.
\newblock Probability inequalities for sums of bounded random variables.
\newblock {\em Journal of the American Statistical Association},
  58(301):13--30, March 1963.

\bibitem{janon2012asymptotic}
A.~Janon, T.~Klein, A.~Lagnoux, M.~Nodet, and C.~Prieur.
\newblock Asymptotic normality and efficiency of two sobol index estimators.
\newblock Preprint available at \texttt{http://hal.inria.fr/hal-00665048/en},
  2012.

\bibitem{liu2003estimating}
R.~Liu and A.~B. Owen.
\newblock {\em Estimating mean dimensionality}.
\newblock Department of Statistics, Stanford University, 2003.

\bibitem{lounici2008sup}
K.~Lounici.
\newblock Sup-norm convergence rate and sign concentration property of lasso
  and dantzig estimators.
\newblock {\em Electronic Journal of statistics}, 2:90--102, 2008.

\bibitem{Monod2006}
H.~Monod, C.~Naud, and D.~Makowski.
\newblock Uncertainty and sensitivity analysis for crop models.
\newblock In D.~Wallach, D.~Makowski, and J.~W. Jones, editors, {\em Working
  with Dynamic Crop Models: Evaluation, Analysis, Parameterization, and
  Applications}, chapter~4, pages 55--99. Elsevier, 2006.

\bibitem{morris1991factorial}
M.~D. Morris.
\newblock Factorial sampling plans for preliminary computational experiments.
\newblock {\em Technometrics}, 33(2):161--174, 1991.

\bibitem{saltelli2008global}
A.~Saltelli, M.~Ratto, T.~Andres, F.~Campolongo, J.~Cariboni, D.~Gatelli,
  M.~Saisana, and S.~Tarantola.
\newblock {\em Global sensitivity analysis: the primer}.
\newblock Wiley Online Library, 2008.

\bibitem{sobol1993}
I.~M. Sobol.
\newblock Sensitivity estimates for nonlinear mathematical models.
\newblock {\em Math. Modeling Comput. Experiment}, 1(4):407--414 (1995), 1993.

\bibitem{sobol2001global}
I.M. Sobol.
\newblock {Global sensitivity indices for nonlinear mathematical models and
  their Monte Carlo estimates}.
\newblock {\em Mathematics and Computers in Simulation}, 55(1-3):271--280,
  2001.

\bibitem{tibshirani1996regression}
R.~J. Tibshirani.
\newblock Regression shrinkage and selection via the lasso.
\newblock {\em Journal of the Royal Statistical Society. Series B
  (Methodological)}, pages 267--288, 1996.

\bibitem{tissot2012bias}
J.-Y. Tissot and C.~Prieur.
\newblock Bias correction for the estimation of sensitivity indices based on
  random balance designs.
\newblock {\em Reliability Engineering \& System Safety}, 107:205--213, 2012.

\bibitem{tissot2012estimating}
J.-Y. Tissot and C.~Prieur.
\newblock Estimating sobol'indices combining monte carlo estimators and latin
  hypercube sampling.
\newblock 2012.

\bibitem{tropp2006just}
J.~A. Tropp.
\newblock Just relax: Convex programming methods for identifying sparse signals
  in noise.
\newblock {\em Information Theory, IEEE Transactions on}, 52(3):1030--1051,
  2006.

\bibitem{welch1974lower}
L.~Welch.
\newblock Lower bounds on the maximum cross correlation of signals (corresp.).
\newblock {\em Information Theory, IEEE Transactions on}, 20(3):397--399, 1974.

\bibitem{xu2007further}
W.~Xu and B.~Hassibi.
\newblock Further results on performance analysis for compressive sensing using
  expander graphs.
\newblock In {\em Signals, Systems and Computers, 2007. ACSSC 2007. Conference
  Record of the Forty-First Asilomar Conference on}, pages 621--625. IEEE,
  2007.

\bibitem{zhao2006model}
P.~Zhao and B.~Yu.
\newblock On model selection consistency of lasso.
\newblock {\em The Journal of Machine Learning Research}, 7:2541--2563, 2006.

\end{thebibliography}

\end{document}